\documentclass[sigconf]{acmart}

\AtBeginDocument{%
  \providecommand\BibTeX{{%
    \normalfont B\kern-0.5em{\scshape i\kern-0.25em b}\kern-0.8em\TeX}}}

\copyrightyear{2022}
\acmYear{2022}
\setcopyright{acmcopyright}
\acmConference[KDD '22] {Proceedings of the 28th ACM SIGKDD Conference on Knowledge Discovery and Data Mining}{August 14--18, 2022}{Madrid, Spain.}
\acmBooktitle{Proceedings of the 28th ACM SIGKDD Conference on Knowledge Discovery and Data Mining (KDD '22), August 14--18, 2022, Madrid, Spain}
\acmPrice{15.00}
\acmISBN{978-1-4503-9385-0/22/08}
\acmDOI{10.1145/3534678.3539468}

\usepackage{amsmath,amsfonts,bm}

\def\sref#1{\S~\ref{#1}}

\def\eqref#1{equation~\ref{#1}}
\def\Eqref#1{Eq.(\ref{#1})}

\def\1{\bm{1}}

\def\rc{{\textnormal{c}}}

\def\ro{{\textnormal{o}}}
\def\rp{{\textnormal{p}}}

\def\rt{{\textnormal{t}}}

\def\rx{{\textnormal{x}}}

\def\rc{{{C}}}

\def\ro{{{O}}}

\def\rp{{{P}}}

\def\rt{{{T}}}

\def\rx{{{X}}}

\def\pr{{\text{Pr}}}

\def\vmu{{\bm{\mu}}}
\def\vsigma{{\bm{\sigma}}}

\def\vo{{\textbf{o}}}

\def\vr{{\textbf{r}}}
\def\vs{{\bm{s}}}
\def\vt{{\textbf{t}}}

\def\vv{{\textbf{v}}}

\def\vx{{\textbf{x}}}

\def\mA{{\bm{A}}}

\DeclareMathAlphabet{\mathsfit}{\encodingdefault}{\sfdefault}{m}{sl}
\SetMathAlphabet{\mathsfit}{bold}{\encodingdefault}{\sfdefault}{bx}{n}

\def\vrx{{\bm{\rx}}}
\def\vrt{{\bm{\rt}}}
\def\vro{{\bm{\ro}}}

\usepackage{enumitem}
\usepackage{hyperref}
\usepackage[capitalize,noabbrev]{cleveref}
\usepackage{booktabs} 
\usepackage{graphicx}
\usepackage{nicefrac}
\usepackage{subfigure}
\usepackage{dsfont}
\usepackage{bm}
\usepackage{url}
\usepackage{bigstrut}
\usepackage[ruled,vlined]{algorithm2e}
\usepackage{multirow}
\usepackage{mathtools, nccmath}
\usepackage{balance}

\newtheorem{theorem}{Theorem}

\newcommand{\stitle}[1]{\vspace{2mm} \noindent {\bf #1}}
\newcommand{\method}[1]{\textsf{#1}}
\newcommand{\model}{\method{Vectorization}{}}

\begin{document}

\title{Scalar is Not Enough: Vectorization-based \\ Unbiased Learning to Rank}
\renewcommand{\shorttitle}{Scalar is Not Enough: Vectorization-based Unbiased Learning to Rank}

\author{Mouxiang Chen}
\affiliation{%
  \institution{Zhejiang University \&  Alibaba-Zhejiang University Joint Institute of Frontier Technologies}
  \country{}
}
\email{chenmx@zju.edu.cn}

\author{Chenghao Liu}
\authornote{
Corresponding authors.
}
\affiliation{
  \institution{Salesforce Research Asia}
  \country{Singapore}
}
\email{chenghao.liu@salesforce.com}

\author{Zemin Liu}
\affiliation{%
  \institution{Singapore Management University}
  \country{Singapore}
}
\email{zmliu@smu.edu.sg}

\author{Jianling Sun}
\authornotemark[1]
\affiliation{
  \institution{Zhejiang University \&  Alibaba-Zhejiang University Joint Institute of Frontier Technologies}
  \country{}
}
\email{sunjl@zju.edu.cn}





\renewcommand{\shortauthors}{Mouxiang Chen et al.}

\begin{abstract}

Unbiased learning to rank (ULTR) aims to train an unbiased ranking model from biased user click logs. Most of the current ULTR methods are based on the examination hypothesis (EH), which assumes that the click probability can be factorized into two scalar functions, one related to ranking features and the other related to bias factors. Unfortunately, the interactions among features, bias factors and clicks are complicated in practice, and usually cannot be factorized in this independent way. Fitting click data with EH could lead to model misspecification and bring the approximation error.

In this paper, we propose a vector-based EH and formulate the click probability as a dot product of two vector functions. This solution is complete due to its universality in fitting arbitrary click functions. Based on it, we propose a novel model named \model\ to adaptively learn the relevance embeddings and sort documents by projecting embeddings onto a base vector. Extensive experiments show that our method significantly outperforms the state-of-the-art ULTR methods on complex real clicks as well as simple simulated clicks. \footnote{Codes are provided at \url{https://github.com/Keytoyze/Vectorization}}

\end{abstract}

\ccsdesc[500]{Information systems~Learning to rank}

\keywords{learning to rank, unbiased learning to rank, examination hypothesis}

\maketitle

\section{Introduction}

Learning to rank (LTR) with click data has been widely employed in modern information retrieval systems since this logged feedback reflects the utility of each document for each user \cite{joachims2005accurately}, which is relatively easy to obtain on a large scale. However, they inherently contain a lot of bias from user behavior \cite{joachims2007evaluating}. For example, users are more likely to observe documents at a higher position, known as position bias, which causes clicks to be biased with the position. Using unbiased learning to rank (ULTR) to remove these biases has attracted increasing research interest \cite{agarwal2019general,joachims2017unbiased}. The key idea is the \textbf{examination hypothesis} (EH): each document has a certain probability of being observed and is then clicked based on the relevance, where the observation depends on some bias factors (e.g., position), and the relevance depends on the features that encoding query and document. The EH can be written as:

\begin{align*}
    P(\text{click})=P(\text{observation}\mid \text{bias factors})\cdot P(\text{relevance}\mid \text{features}).
\end{align*}

\begin{figure}[h]
    \centering
    \includegraphics[width=0.3\textwidth]{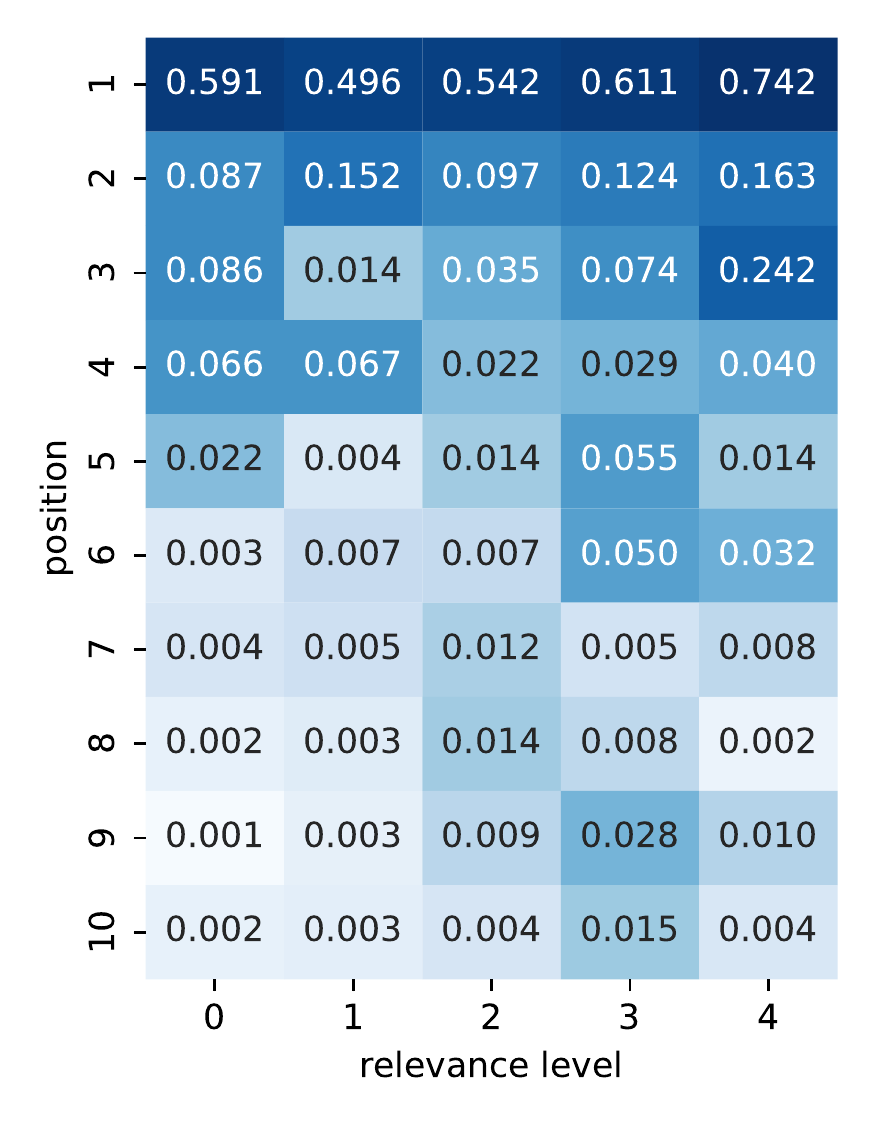}
    \caption{Real click rate matrix calculated on TianGong-ST. Darker colors indicate larger click rates.}
    \label{fig:matrix_click}
\end{figure}

Given this, ULTR methods try to model the observation probability using bias factors, and reweigh the click signals based on the reciprocal of observation probabilities, to recover an unbiased relevance probability for the ranking objective. 

The EH relies on an implicit assumption - the click probability can be factorized into two \textit{scalar} functions, one takes the bias factors as input only and the other takes features only. Unfortunately, this assumption isn't sufficient in practice, since the interaction between relevance and observation is rather complicated \cite{zheng2019constructing,agarwal2019addressing,vardasbi2020inverse}. 


To illustrate it more intuitively, we performed a simple statistical analysis on the TianGong-ST dataset\footnote{\url{http://www.thuir.cn/tiangong-st/}} \cite{chen2019tian}. This dataset contains both user click data (sampled from a real-world search engine) and relevance levels (annotated by humans). We grouped documents by their positions and relevance levels, and compute the click rates for each group as a click rate matrix. Figure \ref{fig:matrix_click} demonstrates the results. We assume the observation depends only on position \cite{joachims2017unbiased,ai2018unbiased,wang2018position,agarwal2019addressing}, and the relevance depends on the relevance level. 
Particularly, if the EH applies to it, then the matrix can be factorized into a $10\times 1$ vector (denotes observation probability for each position) and a $1\times 5$ vector (denotes relevance probability for each level)
, which infers that the matrix rank must be $1$. Obviously, this matrix does not satisfy this condition, since the singular values of this matrix are $(1.40, 0.15, 0.07, 0.06, 0.03)$. It shows that EH cannot describe real-world click rates accurately.


We identify the root cause of this problem to be the fact that the EH is \textit{incomplete}: the click probability can be arbitrary functions related to bias factors and features due to the complicated interaction between observation and relevance, but the function family produced by the combination of such two scalar functions cannot cover all possible functions.
Using this form to fit the click data could lead to model misspecification and bring approximation error no matter how much data we collect. 
Recent efforts extended EH and explicitly described the generative process of clicks in their specific scenarios \cite{agarwal2019addressing,vardasbi2020inverse,ovaisi2020correcting}, which, however, require prior knowledge and still cannot obtain the best performance in common click scenarios due to their insufficiency to cover all possible click functions.

To address this issue, we extend the EH into a vector-based formulation: the click rates can be written as the dot-product of two \textit{vector} functions, one related to bias factors (named as observation embedding), and the other related to features (named as relevance embedding). Moreover, the universality of this factorization can be justified \cite{kaddour2021causal}: for any given click rate function, we can always find an appropriate dimension for these vectors such that the approximation error can be arbitrarily small. This suggests that our vector-based EH is complete and can catch  
complicated click patterns \textit{adaptively} by optimization. Compared to traditional relevance scalars, relevance embeddings have a more powerful capacity to encode how relevant a document is.

However, unlike relevance scalars, embedding vectors cannot be sorted, which challenges the usage of this vector-based EH. To sort the documents with their relevance embeddings in the inference stage, we propose to use a common \textit{base vector} and project each relevance embedding onto it to obtain relevance scalars, and sort the documents with these scalars. The challenge is how to find such a proper common base vector. We argue that we can find the most probable bias factors that ever appear in the training dataset together with all of these ranking features, and use the corresponding observation embedding as the base vector. This is because the dot-product of such a base vector and each relevance embedding is more closed to the real click rate when the features are assigned with the same bias factors. This product result can serve as a substitute for relevance. Given that some bias factors and ranking features may not overlap in the dataset, we further propose to use the most probable observation embedding that ever appears together with features as the base vector. Finally, we derive a closed-form of the base vector, which enables us to calculate it in a very efficient way. 

Moreover, to evaluate the performance of our model in practice, we propose a method to apply the click pattern in the real world in semi-synthetic experiments. Extensive experiments conducted on two widely-used datasets showed that our \model\ method significantly outperforms the state-of-the-art ULTR methods in both simple and complex click settings.

To the best of our knowledge, we are the first to study the limitation of scalar-based EH used by the current ULTR framework. The main contributions of this work are three-fold:
\begin{enumerate}[leftmargin=*]
    \item We propose a vector-based examination hypothesis (vector-based EH) that can capture complicated interactions between clicks, features, and bias factors through a dot-product between relevance embeddings and observation embeddings. This hypothesis is complete for real-world click data.
    \item We propose a method that can sort documents with relevance embeddings, in which each relevance embedding is projected onto a base vector. The base vector can be calculated efficiently.
    \item We provide a method to apply the real click pattern into semi-synthetic experiments.
\end{enumerate}
\section{Related Work}

\stitle{Debiasing Click Data.} 
Most of the current approaches to debiasing click data for ranking are based on the examination hypothesis. They can be divided into two groups. The first is to model user's behavior to infer relevance from biased click signals, known as \textit{click models} \cite{ubm,dcm,ccm,csm,ncm}. However, most click models focus on predicting clicks, and the relevance inference is an afterthought \cite{ai2018unbiased}. The second group tries to directly learn unbiased ranking models from biased clicks, known as \textit{unbiased learning to rank} (ULTR). Based on it, Joachims et al. \cite{joachims2017unbiased} proposed the inverse propensity scoring (IPS) method to reweigh the click signals based on the reciprocal of observation probabilities (called propensity scores) and provide an unbiased estimate of the ranking objective. The propensity scores are estimated by randomized experiments \cite{joachims2017unbiased,wang2016learning}, which hurts users' experience, unfortunately. To address it, Agarwal et al. \cite{agarwal2019estimating} and Fang et al. \cite{fang2019intervention} proposed to do intervention harvest by exploiting click logs with multiple ranking models. Nevertheless, they have a relatively narrow scope of application due to the strict assumption to construct interventional sets \cite{chen2021adapting}. Recently, some researchers proposed to jointly estimate relevance and bias \cite{wang2018position,ai2018unbiased,hu2019unbiased,jin2020deep}. Similar to them, our proposed method could jointly train the ranking model and observation model without intervention. 

On the other side, researchers developed models to extend the scope of bias factors. The bias factors contain position \cite{wang2018position,ai2018unbiased,hu2019unbiased,cai2020debiasing}, contextual information \cite{fang2019intervention,tian2020counterfactual}, clicks in the same query list \cite{chen2021adapting,vardasbi2020cascade}, presentation style \cite{zheng2019constructing,liu2015influence}, search intent \cite{sun2020eliminating} and result domain \cite{ieong2012domain}. In our work, we don't limit the exact meaning of bias factors, which makes our model more flexible and generic.

\stitle{Clicks beyond examination hypothesis.}
There is much work finding that the click functions of features and bias factors are complicated, and cannot be written in the form of scalar-based EH. For instance, in trust bias \cite{joachims2005accurately,agarwal2019addressing,vardasbi2020inverse}, users are more likely to click incorrectly on higher-ranked items, and the relevance scalar function requires an affine transformation about the position, to fit the clicks. As we will mention in this paper, trust bias can be written explicitly as a 2-dimensional vector-based EH. 

Beyond trust bias, there exist other click patterns that cannot be written in the form of EH. Williams et al. \cite{williams2016detecting} and Zheng et al. \cite{zheng2019constructing} argued that some documents with low click necessity will lower the click probability, while the click necessity is related to relevance and bias factors (like presentation style). Liu et al. \cite{liu2014skimming} found that users may examine results in several stages, and different bias factors and features take effect in the different stages. Even though the click function in these scenarios may not be written explicitly as a vector dot-product, the vector-based EH can approximate it thanks to its universality.

\stitle{Vector-based factorization.} Vector-based factorization is widely used in the field of recommendation systems, known as matrix factorization, where a user-item rating matrix is approximated by the product of two low-rank matrices (latent factor vectors) \cite{mnih2007probabilistic,liu2017collaborative,koren2022advances,wang2017interactive,he2016fast}. Besides, \cite{kaddour2021causal} proved the universality of product effect, which is a theoretical guarantee for our vector-based EH. However, this work didn't provide a solution to sort the vectors, which is crucial in the LTR task.

\section{Preliminaries}

In this paper, we use bold letters to denote vectors (e.g., $\vr$), and thin letters to denote scalars (e.g., $r$). Generally, the core of LTR is to learn a ranking model $f$ which assigns a relevance score to a document with its ranking feature. For a query, documents can be sorted in descending order by their scores. In the full information setting that we already know the true relevance for each document, the observational data related to a query $q\in \mathcal{Q}$ can be notated as $\mathcal{D}_q^{\text{full-info}}= \{(\vx_i, r_i)\}_{i=1}^n$, where $\vx_i \in \mathcal{X}$ denotes the ranking features encoding query, document and user, and $r_i\in \mathbb R$ denotes its true relevance score. The ranking target is to optimize $f$ by minimizing the empirical risk:
\begin{align*}
    \mathcal R=\frac{1}{|\mathcal{Q}|} \sum_{q\in\mathcal{Q}} \sum_{(\vx_i,r_i)\in \mathcal{D}_q^{\text{full-info}}} L(f(\vx_i), r_i),
\end{align*}
where $L$ denotes a loss function based on any specific IR metric of interest \cite{joachims2017unbiased}. The true relevance score $r_i$ denotes how relevant a document with the query related to the ranking features $\vx_i$, which is typically obtained by human annotation.

In practice, the relevance scores are often unknown and are costly to estimate through human labeling \cite{chapelle2011yahoo}. Instead, offline Unbiased Learning to Rank (ULTR) methods try to learn the ranking model from offline click logs, which are cheap and timely to obtain at scale. This is because click logs can be seen as implicit feedback which reflects users' preferences to some extent. Nevertheless, click logs are often biased. For example, higher-ranked documents are more likely to be observed and clicked (known as position bias). 

In this click setting, the observational data related to $q$ can be notated as $\mathcal{D}_q=\{(\vx_i,\vt_i,c_i)\}_{i=1}^n$, where $c_i\in\{0,1\}$ denotes the click signals of $\vx_i$, and $\vt_i\in\mathcal{T}$ denotes bias factors that cause clicks to be biased, such as document position \cite{joachims2017unbiased}, context information \cite{fang2019intervention}, other clicks around the document \cite{vardasbi2020cascade,chen2021adapting} or the presentation style \cite{liu2015influence}. In this work, we do not limit the exact meaning of $\vt_i$, which allows us to generalize our conclusion to most of the previous ULTR methods. For convenience, let $\mathcal D=\{(\vx,\vt)\mid (\vx,\vt,c)\in \mathcal{D}_q,q\in\mathcal{Q}\}$ denote all pairs of ranking features and bias factors that ever appear in the dataset.

We assume that the click rate of a document only depends on its ranking features and its bias factors. Denote $c(\vx,\vt)=\pr(c=1\mid \vx,\vt)$ as the click rate function, with $(\vx,\vt)\in\mathcal{D}$. In order to derive relevance from click data, most of the current ULTR methods \cite{joachims2017unbiased,ai2018unbiased,wang2018position,fang2019intervention,jin2020deep,chen2021adapting} are based on \textbf{examination hypothesis} (EH) to model user's click behavior. It assumes that the user clicks on a document if this document is observed and relevant. If we further assume that the relevance $r$ depends on the ranking features $\vx$ and observation $o$ depends on the bias factors, we have:

\begin{align}
    \label{eq:scalar_eh}
    c(\vx,\vt)=r(\vx)\cdot o(\vt),\quad \forall(\vx,\vt)\in\mathcal{D},
\end{align}
where $r(\vx)$ denotes the probability of relevant and $o(\vt)$ denotes the probability of being observed by user. Both $r: \mathcal{X}\rightarrow \mathbb R$ and $o:\mathcal{T}\rightarrow \mathbb R$ are scalar functions. By explicitly modeling the bias effect via observation probability, it is able to attain an unbiased estimate of the ranking objective.
\section{Vectorization-based ULTR}

\begin{figure}[h]
    \centering
    \includegraphics[width=0.4\textwidth]{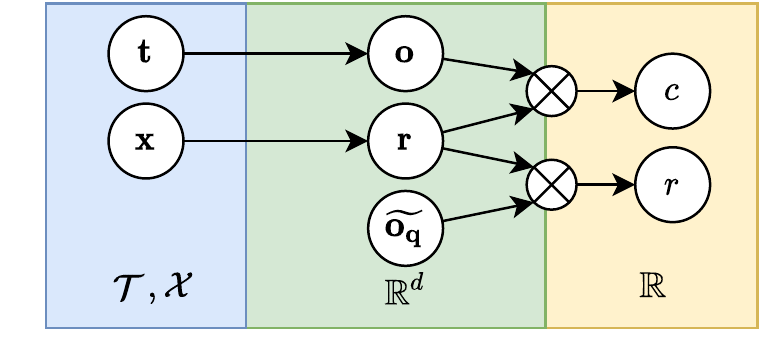}
    \caption{Graph representation of our method.}
    \label{fig:graph}
\end{figure}

In this section, we propose our \model\ to deal with complicated click patterns. Figure \ref{fig:graph} illustrates our framework. We first map bias factors and features into vectors and combine them onto clicks for training (\sref{sec:vec_eh}). For final ranking, we use a common base vector to project relevance embeddings onto scalar (\sref{sec:rank_with_emb}). Finally, we describe how to find such a base vector (\sref{sec:find_base}). In addition, we further discuss the relation with trust bias in \sref{sec:relation_trust}.

\subsection{Vector-based EH}

\label{sec:vec_eh}

Unfortunately, the interaction between clicks, bias factors and features is extremely complicated in real-world \cite{liu2014skimming,liu2015influence,williams2016detecting,buscher2010good}. An ideal factorization for the click in \Eqref{eq:scalar_eh} often does not exist in practice, since the function family produced by this form cannot cover all possible click rate functions. It leads to model misspecification and brings approximation error no matter how much data we collect.

To capture the complicated interaction between relevance and observation, we first extend the scalar-based EH to a vector-based formalization. That is, we assume a \textit{product effect}: the click function $c(\vx,\vt)$ can be written as a dot product of two functions, one over the ranking features $\vx$ and the other over the bias factors $\vt$.

\begin{align}
    \label{eq:vec_eh}
    c(\vx, \vt)=\vr(\vx)^\top\vo(\vt),
\end{align}
where $\vr: \mathcal{X}\rightarrow \mathbb R^d$ is the relevance function and $\vo: \mathcal{T}\rightarrow \mathbb R^d$ is the observation function. The outputs of them are referred to as \textbf{relevance embedding} and \textbf{observation embedding}. Note that, \Eqref{eq:scalar_eh} is a special case of \Eqref{eq:vec_eh} when $d=1$.

The universality of product effect can be formally justified: if we increase the dimensionality $d$ of $\vr$ and $\vo$, any arbitrary bounded continuous function in $c(\mathcal{X}\times \mathcal{T})$ can be approximated. A formal proof of the universality of product effect is provided in the Appendix \sref{sec:university_pe}. It enables us to decompose the biased click into an unbiased part (relevance) and a biased part (observation), no matter how complicated the interaction between relevance and observation is. 

\subsection{Rank with relevance embedding}

\label{sec:rank_with_emb}

\def\vbase{{\widetilde{\vo_q}}}

However, it's not possible to apply the vector-based EH immediately by reason that relevance embedding cannot be sorted according to their values. We need to find a method that can rank documents with their relevance embeddings. For a given query $q$ with $n$ ranking features $\vx_1,\vx_2,\cdots,\vx_n$, our goal is to rank these features with their relevance embeddings $\vr(\vx_1),\cdots,\vr(\vx_n)$. Obviously, it is not suitable to simply average the elements in vectors and sort all vectors based on the average values\footnote{This is because we do not impose a non-negative constraint on the embeddings. We can flip the signs of the relevance embedding and the observation embedding at the same time without changing their product, but the average of relevance embeddings is quite different after flipping. }. 
Consider the form of \Eqref{eq:vec_eh}, a natural solution is to find a common \textbf{base vector} $\vbase \in\mathbb R^d$ for a query $q$, and project each relevance embedding onto $\vbase$:

\begin{align}
    \label{eq:project}
    r(\vx_i)=\vr(\vx_i)^\top \vbase,\quad i\in[n].
\end{align}

Then we can sort with the scalar $r(\vx_i)$, like traditional LTR methods. In fact, for a set of relevance embeddings, there exists a base vector that can sort them in any given order if we allow $d$ to grow, which shows the universality of this projection method. This is because when $r(\vx_i)$ and $\vr(\vx_i)$ are fixed, \Eqref{eq:project} can be seen as linear equations in which $\vbase$ is a variable, $\vr(\vx_i)$ is a coefficient and $r(\vx_i)$ is a constant. Suppose that $\{\vr(\vx_i):i\in[n]\}$ are linearly independent. If $d\geq n$, then the linear equations must have a solution of $\vbase$. It suggests that if the dimension $d$ is large enough, we can always find a base vector such that $r(\cdot)$ can equal arbitrary values.

Now, the next problem is how to find such a base vector. If there exists a few labeled data, it can be done by solving the above linear equations. In this paper, we focus on a general case that the labeled data is unknown. 

\subsection{Find the base vector}

\label{sec:find_base}

In this section, we suppose the function $\vr(\cdot)$ and $\vo(\cdot)$ are given and fixed, and we aim to find the base vector \textit{unsupervisedly}. To start with, we assume that for any two documents, their click rate order equals to their relevance order if we fix their bias factors:

\begin{align}
    \label{eq:assumption_order}
    c(\vx_1,\vt)\geq c(\vx_2,\vt) \iff \vx_1\succeq_{r}\vx_2,\quad\forall (\vx_1,\vt),(\vx_2,\vt)\in\mathcal{D},
\end{align}
where $\vx_1\succeq_{r}\vx_2$ means that $\vx_1$ is more relevant than $\vx_2$. This is owing to that if two documents have the same bias factors, they are in the same environment, thus a more relevant document tends to receive more clicks.

We start with a toy example. For a query $q$ and the corresponding ranking features $\{\vx_1,\cdots,\vx_n\}$, suppose that there exist common bias factors $\vt$ such that they ever appear together with these features in $\mathcal{D}$. Here, we can just set: 
\begin{align}
    \label{eq:vbase_hard}
    \vbase=\vo(\vt), \quad\text{s.t. } (\vx_i,\vt)\in  \mathcal{D},\forall i\in[n].
\end{align}

This is because $\vr(\vx_i)^\top \vbase=\vr(\vx_i)^\top \vo(\vt)$ indicates the click rate of $\vx_i$ and $\vt$, which reflects the relevance because of \Eqref{eq:assumption_order}. Note that if there exists more than one $\vt$, it's hard to decide which one to use. Thus, we can use maximum likelihood estimation (MLE) to choose the bias factors $t^*$ that maximize the probability of appearing together with $\vx_1,\cdots,\vx_n$. Suppose $\mathcal{D}$ is generated from a joint distribution $P(\vrx,\vrt)$, where $\vrx$ are the ranking features and $\vrt$ are the bias factors. Then we set:

\begin{align}
    \label{eq:vbase_t}
    \vbase=\vo(\vt^*),\quad\text{where }\vt^*=\arg\max_{\vt} \prod_{i=1}^n P(\vrt=\vt\mid\vrx=\vx_i),
\end{align}
and $P(\vrt\mid\vrx)$ can be estimated from $\mathcal D$. Generally, the click rate estimation will be more accurate for the bias factors that have a larger probability of $P(\vrt\mid\vrx)$ \cite{zou2020counterfactual}. Thus we select the \textit{most possible} bias factors related to the ranking features as the base vector $\vbase$, since the combination $\vr(\cdot)^\top\vbase$ can be seen as an accurate click rate. 

However, it may be intractable in practice since common bias factors for all ranking features may \textit{not exist} in the training dataset. For example, we assume $\vt$ denotes the position of documents. Suppose that a ranking feature $\vx_1$ is always assigned to the first position, and another ranking feature $\vx_2$ is always assigned to the second position. No matter how we choose a position $\vt$, there must exist a ranking feature $\vx\in\{\vx_1,\vx_2\}$ such that $P(\vt\mid\vx)=0$, which challenges finding base vector by \Eqref{eq:vbase_t}.

\begin{figure*}[htbp]
    \centering
    \includegraphics[width=0.95\textwidth]{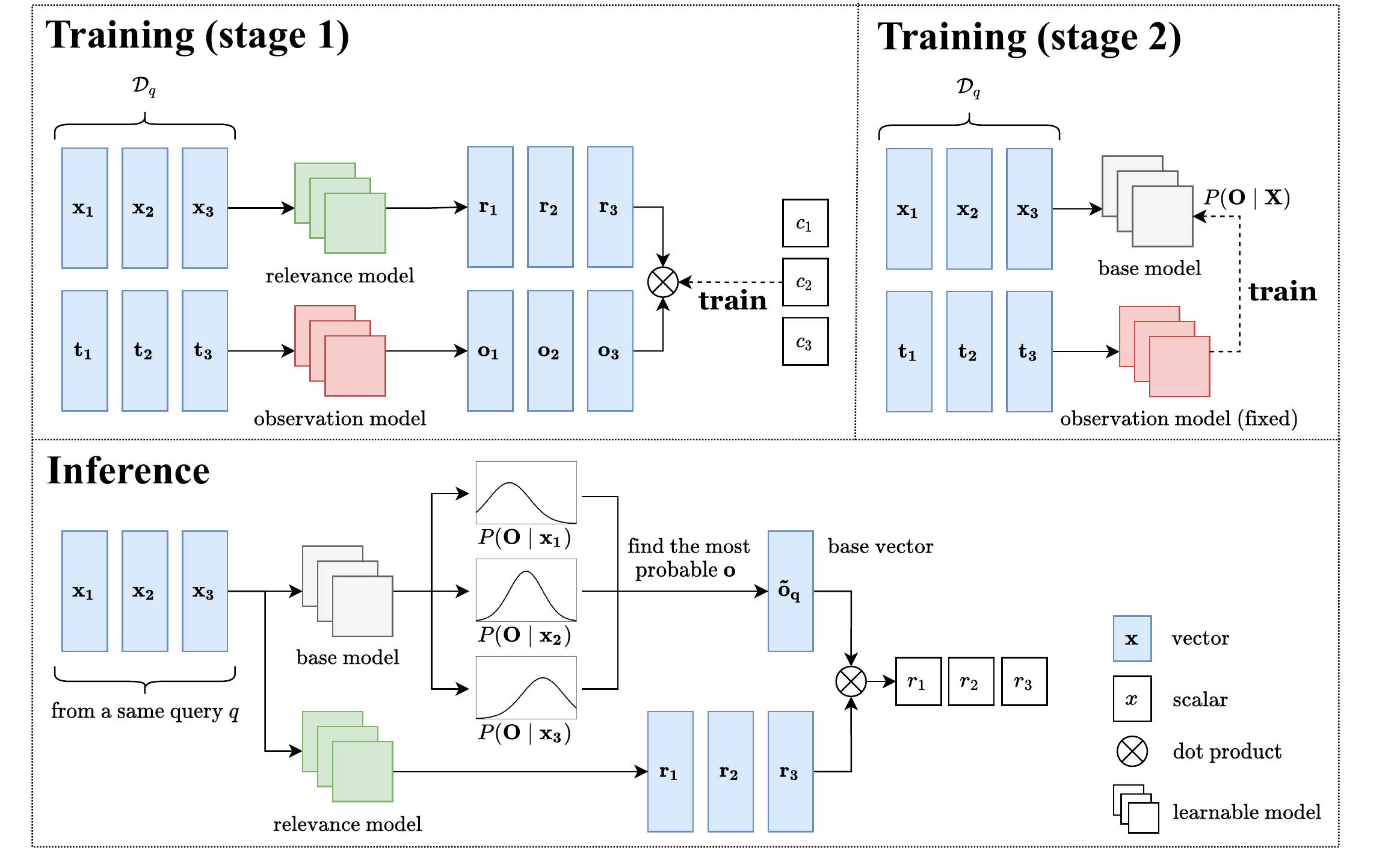}
    \caption{Framework of the proposed \model. In the first training stage, we jointly train a relevance model and an observation model with vector-based EH. In the second training stage, we train a base model to estimate the conditional observation embedding distribution. In the inference stage, we utilize the base model to infer the base vector based on input features and project each relevance embedding onto the base vector for ranking.}
    \label{fig:model}
\end{figure*}

We identify the root cause to be the fact that $\vrx$ and $\vrt$ may not \textit{overlap}: $P(\vrx,\vrt)$ may be zero for some $(\vrx,\vrt)\in\mathcal{X}\times\mathcal{T}$. To address this problem, we first transform $\mathcal{D}$ into $\mathcal{D}_{\vo}=\{(\vx,\vo(\vt)): (\vx,\vt)\in\mathcal{D}\}$. This time we use the observation embedding that maximizes the probability of appearing together with $\vx_1,\cdots,\vx_n$, rather than the raw bias factors:
\begin{align}
    \label{eq:vbase_o}
    \vbase=\arg\max_{\vo}\prod_{i=1}^nP(\vro=\vo\mid \vrx=\vx_i),
\end{align}
where $P(\vro\mid\vrx)$ can be estimated from $\mathcal{D}_{\vo}$. \Eqref{eq:vbase_o} is similar to but better than \Eqref{eq:vbase_t} in use, for that $\vro$ is more dense than the raw bias factors, which alleviates the overlap problem.

\def\vmean{{\vmu(\vx_i)}}
\def\vstd{{\vsigma(\vx_i)}}

Furthermore, we propose to model the $P(\vro\mid\vrx)$ as a multivariate Gaussian distribution (all of its components are independent), since that $P(\vro\mid\vrx)>0$ can always be established, which avoids the overlap problem and make the estimation more stable. More importantly, \Eqref{eq:vbase_o} will have a closed-form solution, which allows us to calculate $\vbase$ in a very efficient way\footnote{We admit that the Gaussian distribution may not be the best model, but we found that it was good enough in the experiment. One explanation is that the relevance embeddings are robust enough to tolerate the small perturbation of the base vector.}. Suppose that $P\left(\vro\mid\vrx=\vx_i\right)\sim \mathcal N\left(\vmean, \vstd^2\right)$, where $\vmean$ and $\vstd$ are the mean and the standard deviation of $\vro$ (given $\vrx=\vx_i$), we have:
\begin{align*}
    \vbase&=\arg\max_{\vo}\prod_{i=1}^nP
    \left(
        \vro=\vo\mid \vrx=\vx_i
    \right)\\
    &=\arg\max_{\vo} \prod_{i=1}^n \frac{1}{\vstd\sqrt{2\pi}}
    \exp\left(-\frac{(\vo-\vmean)^2}{2\vstd^2}\right)\\
    &=\arg\max_{\vo} \sum_{i=1}^n\left(-\frac{(\vo-\vmean)^2}{2\vstd^2}\right),
\end{align*}
where every operators are element-wise. The second equality uses the PDF of Gaussian distribution. Let:
\begin{align*}
    F(\vo)=\sum_{i=1}^n\left(-\frac{(\vo-\vmean)^2}{2\vstd^2}\right),
\end{align*}
by solving the equation $\nicefrac{\text{d}F}{\text{d}\vo}=0$ we reach the final form:

\begin{align}
    \label{eq:vbase_final}
    \vbase=\frac
    {\sum_{i=1}^n \frac{1}{\vstd^2} \vmean}
    {\sum_{i=1}^n \frac{1}{\vstd^2}}.
\end{align}

It indicates that for a given query, the base vector can be calculated through a weighted average over the most probable observation embeddings of all ranking features related to $q$. The larger the variance, the smaller the weight. The variance can be explained as the model's uncertainty. Thus, a major advantage of our method is that it can comprehensively consider the base vector according to the uncertainty, which helps to increase the robustness.

\subsection{Relation with trust bias}

\label{sec:relation_trust}

Most of current ULTR methods assume the scalar-based EH (\Eqref{eq:scalar_eh}), while an exception is the trust bias model \cite{agarwal2019addressing,vardasbi2020inverse}, in which the click rate $c(\vx,\vt)$ with regard to the ranking features $\vx$ and the bias factors $\vt$ (usually take the position) can be written as:
\begin{align*}
    c(\vx,\vt) = \theta(\vt)\left(
        \epsilon^+(\vt) r(\vx) + \epsilon^-(\vt) (1-r(\vx))
    \right),
\end{align*}
where $\epsilon^+$, $\epsilon^-$ and $\theta$ are some functions related to bias factors. This formula cannot be written in the form of scalar-based EH. Oppositely, it can be regarded as a special case of 2-dimensional vector-based EH (\Eqref{eq:vec_eh}), as long as we take:
\begin{align*}
    \vr(\vx)&=[r(\vx),1]^\top,\\
    \vo(\vt)&=\theta(\vt) [\epsilon^+(\vt)- \epsilon^-(\vt), \epsilon^-(\vt)]^\top,
\end{align*}
and always select $\vbase=[1,0]^\top$ as the base vector for final ranking. This proves the ability of our method to deal with trust bias. In the experiment, we find that in the trust bias click setting, when we take $d=2$, the performance of our model is similar to \method{Affine} \cite{vardasbi2020inverse}, which is designed to deal with trust bias specifically. This shows the method we used to find the base vector works.

\section{Model Implementation}

So far, we have shown that to effectively catch the complicated interaction among clicks, bias factors and features, we should learn a relevance embedding that can fully combine with the observation, and project the relevance embedding into the base vector to sort the documents. In this section, we introduce the implementation of the proposed \model\ method, in the training stage and the inference stage. Figure \ref{fig:model} illustrates the overall model architecture.

\subsection{Training stage} 

\stitle{Stage 1.}
We learn two models at first: relevance model $\vr$ and observation model $\vo$. For a query $q$ with the data $\mathcal D_q=\{(\vx_i,\vt_i,c_i)\}_{i=1}^n$, we first combine the relevance embedding and the observation embedding with dot-product, by \Eqref{eq:vec_eh}. Similar to \cite{ai2018unbiased}, we use a list-wise loss based on softmax-based cross entropy:

\begin{align}
    \label{eq:loss_click}
    L_{q}^{click}=-\sum_{i=1}^n c_{i} \log \frac{
            \exp\left(\vr(\vx_i;\theta_r)^\top \vo(\vt_i;\theta_o)\right)
        }{
            \sum_{j=1}^n  \exp\left(\vr(\vx_j;\theta_r)^\top \vo(\vt_j;\theta_r)\right)
        },
\end{align}
where $\theta_r$ and $\theta_o$ are the weights of the relevance model and the observation model respectively. The softmax-based cross entropy naturally converts the combinational output $\vr(\cdot)^\top\vo(\cdot)$ into click probability distributions. 

\stitle{Stage 2.}
After the convergence of the observation model, we need to fix it and learn a base model $\vv$ that estimates the distribution $P(\vro\mid\vrx)$ to find the base vector for the inference stage. As mentioned above, we fix a Gaussian likelihood to model the distribution, and the output of $\vv$ is composed of both predictive mean as well as predictive variance:
\begin{align*}
    \left[
        \vmu(\vx_i), \vsigma^2(\vx_i)
    \right] = \vv(\vx_i;\theta_v),
\end{align*}
where $\vmu(\vx_i)\in\mathbb R^d$ and $\vsigma^2(\vx_j)\in\mathbb R^d$ are the outputs of the base model $\vv$, and $\theta_v$ is the weight. We want to make the Gaussian distribution parameterized by $\vmu$ and $\vsigma$ close to the real distribution $P(\vro\mid\vrx)$. This can be done by minimizing the following regression loss, according to Kendall et al. \cite{kendall2017uncertainties}:
\begin{align}
    \label{eq:loss_base}
    L_q^{base}=\frac{1}{2}\sum_{i=1}^n \left(
        \frac{
            \left|\left| 
                \vmu(\vx_i)-\vo(\vt_i) 
            \right|\right|_2^2
        }{
            \vsigma^2(\vx_i)
        } + 
        \log \vsigma^2(\vx_i)
        \right) + \lambda || \theta_v ||_2^2,
\end{align}
where $\lambda$ is the hyper-parameter controlling the L2 regularization. In practice, we train the model to predict the log variance, $\vs(\vx_i):=\log \vsigma^2(\vx_i)$, because it is more numerically stable than regressing the variance, as the loss avoids a potential division by zero \cite{kendall2017uncertainties}.

\subsection{Inference stage}

In the inference stage, we discard the observation model $\vo$ and use $\vr$ and $\vv$ to estimate relevance scalars for ranking. For a query $q$ with the ranking features $\{\vx_1,\cdots,\vx_n\}$, we first calculate the base vector $\vbase$ by \Eqref{eq:vbase_final} with the base model, and then project each relevance embedding onto the base vector by \Eqref{eq:project}. 

We further present the algorithm for model
training and inference in Appendix \sref{sec:algorithm}.

\section{Experiments}

In this section, we describe our experimental setup and show the empirical results. We assume that the bias factors only contain the positions, so we modeled the observation model as a position-based model (PBM) in our experiments. One could easily extend it to other models that consider more bias factors. Correspondingly, all of our semi-synthetic setups used positions as the only bias factors.

\subsection{Experimental Setup}

\paragraph{Dataset}

\begin{table*}[htbp]
  \centering
  \caption{Comparison of different algorithms on two datasets in the real click setting. Numbers are shown with their standard deviation in 8 repeated runs (i.e., $\pm$ x).}
    \begin{tabular}{c|cccc|cccc}
    \hline
    \multicolumn{9}{c}{Real Click Setting} \bigstrut\\
    \hline
    \multirow{2}[4]{*}{Algorithms} & \multicolumn{4}{c|}{Yahoo!} & \multicolumn{4}{c}{Istella-S} \bigstrut\\
\cline{2-9}         & nDCG@1 & nDCG@3 & nDCG@5 & nDCG@10 & nDCG@1 & nDCG@3 & nDCG@5 & nDCG@10 \bigstrut\\
    \hline
    \hline
    \method{Labeled Data} & $0.680_{\pm.002}$ & $0.689_{\pm.001}$ & $0.711_{\pm.001}$ & $0.758_{\pm.001}$ & $0.655_{\pm.001}$ & $0.632_{\pm.001}$ & $0.660_{\pm.001}$ & $0.723_{\pm.000}$ \bigstrut[t]\\
    \method{Click Data} & $0.615_{\pm.010}$ & $0.635_{\pm.005}$ & $0.662_{\pm.004}$ & $0.718_{\pm.003}$ & $0.596_{\pm.002}$ & $0.576_{\pm.001}$ & $0.607_{\pm.001}$ & $0.676_{\pm.001}$ \bigstrut[b]\\
    \hline
    \method{DLA}  & $0.625_{\pm.013}$ & $0.643_{\pm.010}$ & $0.671_{\pm.010}$ & $0.725_{\pm.008}$ & $0.599_{\pm.016}$ & $0.580_{\pm.015}$ & $0.610_{\pm.012}$ & $0.679_{\pm.010}$ \bigstrut[t]\\
    \method{PairDebias} & $0.612_{\pm.005}$ & $0.633_{\pm.003}$ & $0.661_{\pm.003}$ & $0.717_{\pm.002}$ & $0.595_{\pm.002}$ & $0.576_{\pm.001}$ & $0.607_{\pm.001}$ & $0.676_{\pm.000}$ \\
    \method{RegressionEM} & $0.618_{\pm.010}$ & $0.628_{\pm.007}$ & $0.654_{\pm.006}$ & $0.709_{\pm.005}$ & $0.574_{\pm.010}$ & $0.544_{\pm.009}$ & $0.565_{\pm.012}$ & $0.622_{\pm.015}$ \\
    \method{Affine} & $0.654_{\pm.003}$ & $0.659_{\pm.003}$ & $0.682_{\pm.003}$ & $0.733_{\pm.003}$ & $0.637_{\pm.005}$ & $0.605_{\pm.004}$ & $0.628_{\pm.005}$ & $0.686_{\pm.006}$ \\
    \hline
    \model & \boldmath{}\textbf{$0.668_{\pm.002}$}\unboldmath{} & \boldmath{}\textbf{$0.674_{\pm.002}$}\unboldmath{} & \boldmath{}\textbf{$0.697_{\pm.002}$}\unboldmath{} & \boldmath{}\textbf{$0.747_{\pm.001}$}\unboldmath{} & \boldmath{}\textbf{$0.643_{\pm.004}$}\unboldmath{} & \boldmath{}\textbf{$0.613_{\pm.003}$}\unboldmath{} & \boldmath{}\textbf{$0.635_{\pm.002}$}\unboldmath{} & \boldmath{}\textbf{$0.694_{\pm.003}$}\unboldmath{} \\
    Increment (vs. \method{DLA}) & 6.88\% & 4.82\% & 3.87\% & 3.03\% & 7.35\% & 5.69\% & 4.10\% & 2.21\% \\
    Increment (vs. \method{Affine}) & 2.14\% & 2.28\% & 2.20\% & 1.91\% & 0.94\% & 1.32\% & 1.11\% & 1.17\% \bigstrut[b]\\
    \hline
    \end{tabular}%
  \label{tab:result_tiangong}%
\end{table*}%

\label{sec:dataset}

We didn't conduct our experiments on the TianGong-ST dataset, since the ranking features are highly limited and the performance of ranking models can be volatile, as reported in \cite{ai2021unbiased}. Instead, we follow the standard semi-synthetic setup in ULTR \cite{ai2018ultra,ai2018unbiased,vardasbi2020inverse} and conduct experiments on two widely used public benchmark datasets: Yahoo! LETOR\footnote{\url{https://webscope.sandbox.yahoo.com/}} \cite{chapelle2011yahoo} and Istella-S\footnote{\url{http://quickrank.isti.cnr.it/istella-dataset/}} \cite{lucchese2016post}. These two datasets have more ranking features than TianGong-ST. More importantly, this enabled us to simulate clicks under different bias settings, including complicated clicks close to the real scene and simple clicks simulated by click models. We provide further details for these datasets in Appendix \sref{sec:ds_detail}.

We followed the data split of training, validation and testing given by the datasets. To generate initial ranking lists for click simulation, we followed the standard process \cite{joachims2017unbiased,ai2018unbiased,chen2021adapting} and use 1\% of the training data with relevance labels to train a Ranking SVM model \cite{joachims2006training}. Based on these initial ranking lists generated by it, we sampled clicks under different click settings.

\paragraph{Click Simulation}

We considered the following two click settings. In both cases, only the top $n=10$ documents were considered
to be displayed. $y_{\max}$ denotes the maximum of the relevance level. For Yahoo! and Istella-S, $y_{\max}=5$.

\textbf{Real Click}. Define $\mA=[a_{ij}]\in\mathbb R^{y_{\max}\times n}$ as the real click matrix, in which each element $a_{ij}$ represents the click rate of a document that locates at the position $i$ and has a relevance level of $j$. The value of $\mA$ is calculated from the TianGong-ST dataset and given in Figure \ref{fig:matrix_click}. We sampled clicks on the two datasets according to this real click matrix $\mA$:
\begin{align}
    \pr(\rc=1\mid \vrx=\vx,\rp=p) = a_{py},
\end{align}
where $y\in[1,y_{\max}]$ is the relevance level. It allowed us to apply the real click rates to our experiments.

\textbf{Trust Bias}. We applied Agarwal et al. \cite{agarwal2019addressing}'s trust bias model to simulate clicks as our second click setting. The click probability can be written as:
\begin{align}
    \pr(\rc=1\mid \vrx=\vx,\rp=p) = \theta_p \left(
        \epsilon_p^+\gamma_{y}+\epsilon_p^-(1-\gamma_{y})
    \right),
\end{align}
where the relevance probability $\gamma_{y}$ is based on the relevance level $y$. We followed previous work \cite{chapelle2009expected,chen2021adapting,ai2018unbiased} and set:
\begin{align}
    \gamma_{y}=\frac{2^{y-1}-1}{2^{y_{\max}-1}-1},
\end{align}
and we used the position bias parameter $\theta_p$ estimated by Joachims et al. \cite{joachims2005accurately}. For the trust bias parameters $\epsilon_p^+$ and $\epsilon_p^-$, we adopted the following formula used by Vardasbi et al. \cite{vardasbi2020inverse}:
\begin{align}
    \epsilon_p^+=1-\frac{p+1}{100},\quad \epsilon_p^-=\frac{0.65}{p}.
\end{align}

Note that trust bias clicks follow 2-dimensional vector-based EH, while real clicks follow 5-dimensional vector-based EH since $\mA$'s rank is $5$.

\paragraph{Baselines}

The baselines consist of the state-of-the-art ULTR methods, including \method{RegressionEM} \cite{wang2018position}, \method{DLA} \cite{ai2018unbiased}, \method{PairDebias} \cite{hu2019unbiased}, \method{Affine} \cite{vardasbi2020inverse}, \method{Labeled Data} (uses human-annotated relevance labels to train the ranker directly, which provides an upper bound of ranker) and \method{Click Data} (uses the raw click data to train the ranker directly). Note that \method{DLA} and \method{RegressionEM} assume scalar-based EH. \method{Affine} is designed to deal with trust bias, which is a special case of 2-dimensional vector-based EH as we discuss in \sref{sec:relation_trust}. Training details can be found in Appendix~\sref{sec:training_detail}.

\subsection{Experimental Results}

\stitle{How does our method perform in the real click setting?} Table \ref{tab:result_tiangong} summarizes the results of the performance on the two datasets, where clicks are simulated according to the TianGong-ST real click matrix.  Particularly, we have the following findings:

\begin{figure*}[htbp]
    \centering
    \subfigure[real click (Yahoo!)]{
        \includegraphics[width=0.3\textwidth]{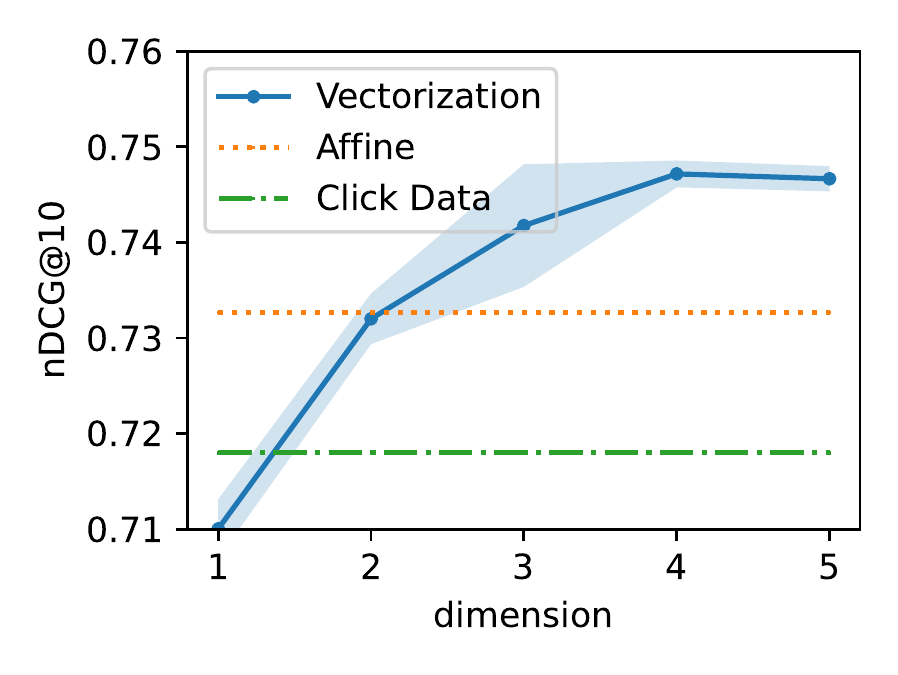}
        \label{fig:dimension_yahoo_tiangong}
    }
    \subfigure[real click (Istella-S)]{
        \includegraphics[width=0.3\textwidth]{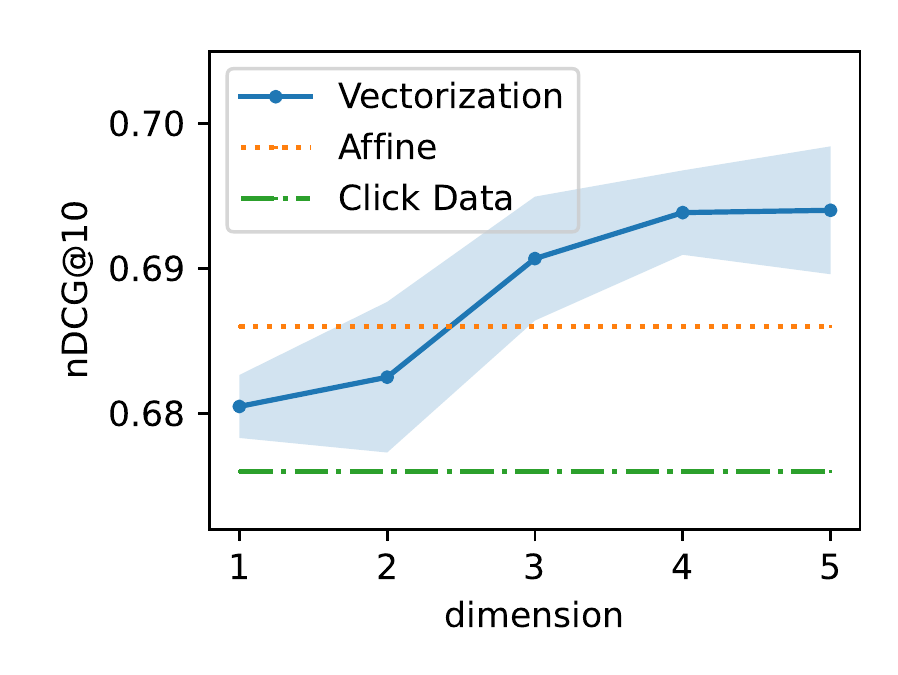}
        \label{fig:dimension_istella_tiangong}
    }
    \subfigure[trust bias (Yahoo!)]{
        \includegraphics[width=0.3\textwidth]{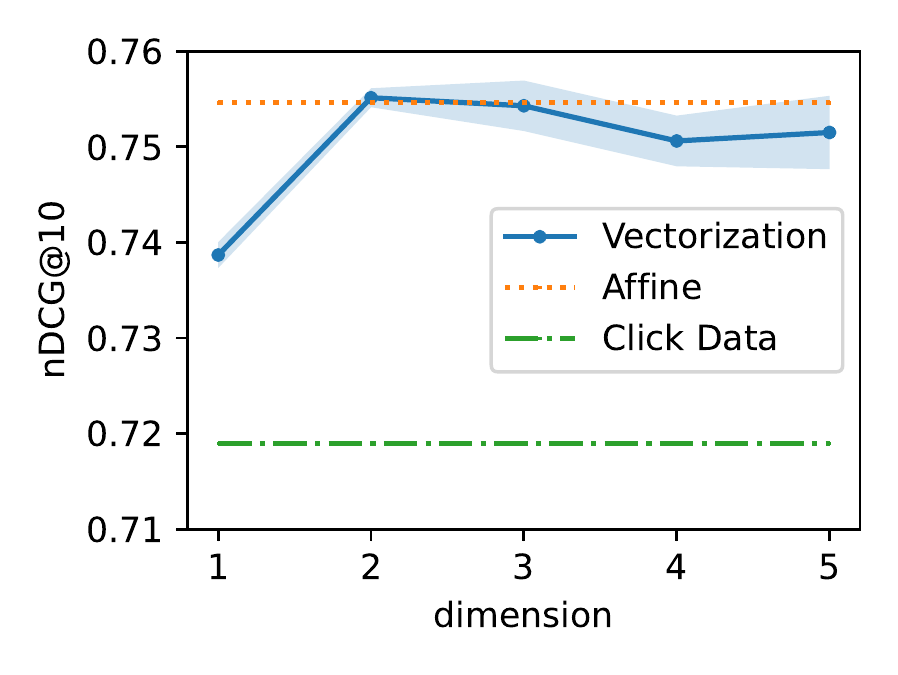}
        \label{fig:dimension_yahoo_trust}
    }
    \caption{Performance changes after adjusting the dimension in \model, where the figure titles are in the form of "click setting (dataset)". The variance is displayed with the shadow areas.}
    \label{fig:dimension}
\end{figure*}
\begin{figure*}[htbp]
    \centering
    \subfigure[real click (Yahoo!)]{
        \includegraphics[width=0.3\textwidth]{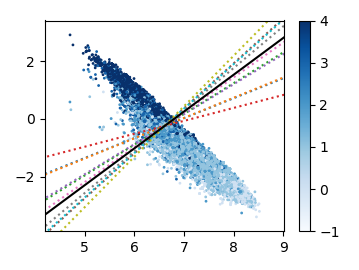}
        \label{fig:scatter_tiangong}
    }
    \subfigure[trust bias (Yahoo!)]{
        \includegraphics[width=0.3\textwidth]{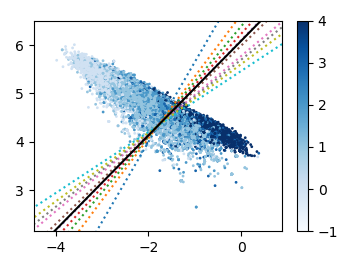}
        \label{fig:scatter_trust}
    }
    \subfigure[real click (Yahoo!)]{
        \includegraphics[width=0.3\textwidth]{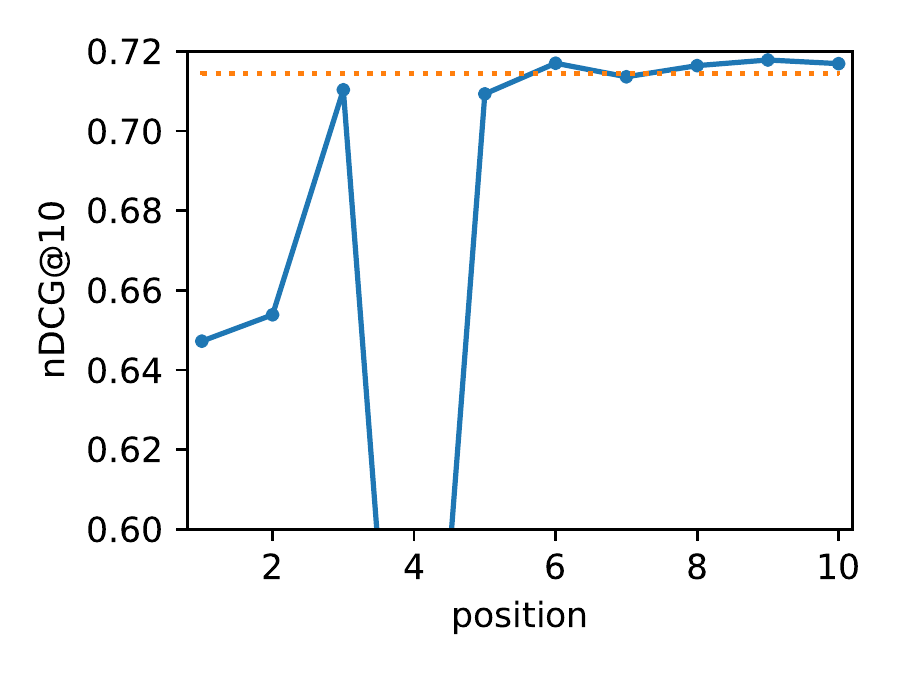}
        \label{fig:performance_position_yahoo_tiangong}
    }
    \caption{(a)(b) Visualization of the embedding space. Points represent relevance embeddings, and their colors represent the relevance level. Directions of observation embeddings for each position are marked as colorful dotted lines, and the average directions of base vectors are marked as black solid lines. (c) Performance when projecting relevance embeddings in Figure \ref{fig:scatter_tiangong} onto observation embeddings for each position (solid line), and when projecting the one onto base vectors (dotted line).}
    \label{fig:scatter}
\end{figure*}

\begin{enumerate}[leftmargin=*]
    \item Our method achieves better performance than all the state-of-the-art methods in terms of all measures. It demonstrates that the vector-based EH is more in line with real-world clicks.
    \item \method{Affine} works better than all the other baselines that follow the scalar-based examination hypothesis. Therefore, clicks in the real world are somewhat consistent with the trust bias since the trust bias model considers two-dimensional combinations.
    \item The standard deviation of our model is further less than that of \method{Affine}, which is less than that of \method{DLA} and \method{RegressionEM}. It shows that increasing the combination dimension can reduce the variance and obtain a more stable ranking model.
\end{enumerate}

\begin{table}[ht]
  \centering
  \caption{Comparison of different algorithms on Yahoo! dataset in the trust bias click setting. Numbers are shown with their standard deviation in 8 repeated runs (i.e., $\pm$ x).}
    \begin{tabular}{c|ccc}
    \hline
    \multicolumn{4}{c}{Trust Bias Setting} \bigstrut\\
    \hline
    \multirow{2}[4]{*}{Algorithms} & \multicolumn{3}{c}{Yahoo!} \bigstrut\\
\cline{2-4}         & nDCG@1 & nDCG@5 & nDCG@10 \bigstrut\\
    \hline
    \hline
    \method{Labeled Data} & $0.680_{\pm.002}$ & $0.711_{\pm.001}$ & $0.758_{\pm.001}$ \bigstrut[t]\\
    \method{Click Data} & $0.615_{\pm.008}$ & $0.664_{\pm.003}$ & $0.719_{\pm.002}$ \bigstrut[b]\\
    \hline
    \method{DLA}  & $0.663_{\pm.004}$ & $0.702_{\pm.002}$ & $0.751_{\pm.001}$ \bigstrut[t]\\
    \method{PairDebias} & $0.624_{\pm.005}$ & $0.668_{\pm.001}$ & $0.723_{\pm.001}$ \\
    \method{RegressionEM} & $0.646_{\pm.008}$ & $0.686_{\pm.003}$ & $0.737_{\pm.002}$ \\
    \method{Affine} & $0.678_{\pm.003}$ & \boldmath{}\textbf{$0.707_{\pm.001}$}\unboldmath{} & \boldmath{}\textbf{$0.755_{\pm.001}$}\unboldmath{} \\
    \hline
    \model & \boldmath{}\textbf{$0.679_{\pm.003}$}\unboldmath{} & \boldmath{}\textbf{$0.707_{\pm.001}$}\unboldmath{} & \boldmath{}\textbf{$0.755_{\pm.001}$}\unboldmath{} \\
    Increment (vs. \method{DLA}) & 2.41\% & 0.71\% & 0.53\% \\
    Increment (vs. \method{Affine}) & 0.15\% & 0.00\% & 0.00\% \bigstrut[b]\\
    \hline
    \end{tabular}%
  \label{tab:result_trust}%
\end{table}%

\stitle{Can our method remove trust bias?} Table \ref{tab:result_trust} summarizes the results about the performance when clicks are generated by the trust bias model on Yahoo! dataset. We can observe that the performance of our model is fairly close to that of \method{Affine}, which shows the ability of our method for removing trust bias, and verifies its efficiency in finding the base vector. Both our method and \method{Affine} outperform the EH-based algorithms, which proves once again the limitation of the scalar-based EH. Besides, compared to Table \ref{tab:result_tiangong}, the performance of our method training on trust bias clicks is better than that of training on the real-world clicks. One explanation is that the trust bias model follows \Eqref{eq:assumption_order}: the click rate relative order will not change with the position, which benefits the finding of base vectors. However, the real clicks are complicated and do not follow this assumption. There is still room for debiasing real clicks.

\stitle{How many dimensions does the model need?} We further tune the dimension in our method to verify its impact on performance.
\begin{itemize}[leftmargin=*]
    \item Figure \ref{fig:dimension_yahoo_tiangong} and \ref{fig:dimension_istella_tiangong} demonstrates the results in the real click setting. We can see that the model performance rises as the dimension increases, which shows that a higher dimension can better capture complex click patterns. When the dimension is $2$, the performance of our method is similar to \method{Affine}, which is consistent with the fact that trust bias is a special case of 2-dimensional vector-based EH.
    \item Figure \ref{fig:dimension_yahoo_trust} demonstrates the results in the trust bias click setting. The best dimension is $2$, which shows that a dimension that matches real clicks can get the best performance. Besides,  increasing the dimension will not significantly degrade performance.
\end{itemize}

\stitle{What does our \model\ model learn?} For convenience of visualization, we set the dimension to $2$ and drawed all the relevance embeddings in Figure \ref{fig:scatter_tiangong} and \ref{fig:scatter_trust}. The direction of observation embeddings and the average direction of base vectors are drawn as a straight line. We discuss more property of this embedding space in Appendix \sref{sec:discuss_relevance_emb}. 

\stitle{What if we choose other observation embeddings as the base vector?}  Since there are only ten positions, we can choose observation embeddings for each position as the base vector to see the ranking performance. Figure \ref{fig:performance_position_yahoo_tiangong} shows the result based on Figure \ref{fig:scatter_tiangong}. We can see that the result projected 
onto the observation embedding changes drastically with the change of position, especially at the higher-rank position. As a comparison, the result of projecting them onto the base vector has always been well. It shows that our method of choosing the base vector is robust.

\stitle{How much extra time will be introduced in the inference phase?} Table \ref{table:infer_time} shows the inference time cost comparison of our \model\ with the conventional scalar-based ranker used by baselines. The inference is 21.4\% slower because of the more complicated model structure. Fortunately, since the base model and the ranker can run in parallel, we claim that this gap can be further narrowed through a more refined concurrency programming.

\begin{table}[h]
    \caption{Time spent sorting a batch (=256) of query lists. Each experiment was repeated for 3,000 times on a Tesla K80 GPU.}
    \label{table:infer_time}
    \begin{tabular}{rrr}
    \toprule
              \model  & \method{Scalar-based Ranker} & Increment \\ \midrule
     0.136 s & 0.112 s  & +21.4\% \\
    \bottomrule
    \end{tabular}
\end{table}
\section{Conclusion and future work}

\stitle{Conclusions.}
In this work, we take the first step to studying the limitation of the scalar-based examination hypothesis (EH). To better catch the complicated click pattern in the real world, we propose vector-based EH, in which observation and relevance interact in a higher-dimensional embedding space. Its universality can be justified, which shows that this hypothesis is complete. To rank documents with relevance embeddings in the inference stage, we propose to find the most probable observation embedding that ever appears with the given features in the training dataset as the base vector, and project relevance embeddings onto it. We obtain a close form of the base vector by modeling the observation embedding distribution as Gaussian distribution. We propose a new method to simulate more realistic clicks for testing. Extensive experiments showed that our \model\ significantly outperforms the state-of-the-art ULTR methods on complex real clicks as well as simple simulated clicks.

\stitle{Future work.}
In future work, it would be interesting to investigate better solutions to sort relevance embeddings. Besides, how to directly optimize a certain ranking metric (e.g., nDCG) with the vector-based EH is also an open question.

\begin{acks}
Thanks to Yusu Hong for proofreading the article, and the reviewers for their
valuable comments and suggestions.
\end{acks}

\bibliographystyle{ACM-Reference-Format}
\bibliography{reference}

\newpage
\appendix
\section*{Appendix}
\renewcommand\thesubsection{\Alph{subsection}}
\renewcommand\thesubsubsection{\thesubsection.\arabic{subsection}}

\subsection{Universality of product effect}

\label{sec:university_pe}

We prove that any bounded continuous function on $\mathcal{X}\times\mathcal{T}$ can be approximated with the dot-product of two vector functions on $\mathcal{X}$ and $\mathcal{T}$ respectively:

\begin{theorem}
Let $\mathcal{H}_{\mathcal{X}\times\mathcal{T}}$ be a Reproducing Kernel Hilbert Space
(RKHS) on the set $\mathcal{X}\times\mathcal{T}$ with universal kernel $k$. For any $\delta>0$, and any $f\in\mathcal{H}_{\mathcal{X}\times\mathcal{T}}$, there is a $d\in\mathbb{N}$ such that there exist two $d$-dimensional vector fields $g:\mathcal{X}\rightarrow\mathbb R^d$ and $h:\mathcal{T}\rightarrow\mathbb R^d$, where $||f-g^\top h||_{L_2(P_{\mathcal{X}\times\mathcal{T}})}\leq\delta$.
\end{theorem}

\begin{proof}
The proof can be found in Proposition 1 of \cite{kaddour2021causal}.
\end{proof}

\subsection{Algorithm}

\label{sec:algorithm}

\stitle{Training stage.} We illustrate the model training of \model\ in  Alg.~\ref{algo:training}.
In line 1, we initialize all the parameters. In lines 2-7, we jointly train the relevance model and the observation model through vector-based EH to make their dot product close to the clicks. In lines 8-12, we train the base model, to let the distribution estimation close to the observation embedding distribution. 

\begin{algorithm}
    \caption{\textsc{Model Training for} \model}
    \label{algo:training}
    \LinesNumbered
    \KwIn{dataset $\{\mathcal{D}_q\}_{q\in\mathcal{Q}}$, learning rate $\alpha$, $\lambda$}
    \KwOut{model parameters $\theta_r$, $\theta_o$, $\theta_v$}
    Initialize $\theta_r$, $\theta_o$ and $\theta_v$\; 
    \tcc{Stage 1.}
    \Repeat{$\theta_r$ and $\theta_o$ convergence}{
        sample a query $q\in\mathcal{Q}$, $\{(\vx_i,\vt_i,c_i)\}_{i=1}^n \leftarrow \mathcal{D}_q$\;
        compute $L_q^{click}$ with \Eqref{eq:loss_click}\;
        $\theta_r\leftarrow \theta_r - \alpha \cdot \nicefrac{\partial L_q^{click}}{\partial \theta_r}$\;
        $\theta_o\leftarrow \theta_o - \alpha \cdot \nicefrac{\partial L_q^{click}}{\partial \theta_o}$\;
    }
    \tcc{Stage 2.}
    \Repeat{$\theta_v$ convergence}{
        sample a query $q\in\mathcal{Q}$, $\{(\vx_i,\vt_i,c_i)\}_{i=1}^n \leftarrow \mathcal{D}_q$\;
        compute $L_q^{base}$ with \Eqref{eq:loss_base}\;
        $\theta_v\leftarrow \theta_v - \alpha \cdot \nicefrac{\partial L_q^{base}}{\partial \theta_v}$\;
    }
    \Return{$\theta_r$, $\theta_o$, $\theta_v$}
\end{algorithm}

\stitle{Inference stage.} The overview of the algorithm in the inference stage is summarized in Alg.~\ref{algo:inference}. In lines 1-2, we compute the observation embedding distributions for ranking features. In line 3, we compute the base vector. In lines 4-5, we project the relevance embeddings onto the base vector to obtain ranking scores.

\begin{algorithm}
    \caption{\textsc{Model Inference for} \model}
    \label{algo:inference}
    \LinesNumbered
    \KwIn{a set of ranking features $\{\vx_i\}_{i=1}^n$ related to a query, model parameters $\theta_r$ and $\theta_v$}
    \KwOut{ranking scores $r(\vx_1),r(\vx_2),\cdots,r(\vx_n)$}
    \For{$i=1$ \KwTo $n$}{
        $\vmu(\vx_i), \vsigma^2(\vx_i)\leftarrow \vv(\vx_i;\theta_v)$\;
    }
    compute $\vbase$ with \Eqref{eq:vbase_final}\;
    \For{$i=1$ \KwTo $n$}{
        $r(\vx_i)\leftarrow\vr(\vx_i; \theta_r)^\top \vbase$\;
    }
    \Return{$r(\vx_1),r(\vx_2),\cdots,r(\vx_n)$}
\end{algorithm}

\subsection{Further details of datasets}

\label{sec:ds_detail}

Table \ref{table_data_statistics} shows the characteristics of the two datasets, Yahoo! and Istella-S, in addition to the TianGong-ST which we used to estimate the real click rate function. 
\subsection{Training details}

\begin{table}[h]
    \caption{Dataset statistics}
    \label{table_data_statistics}
    \begin{tabular}{lrrr}
    \toprule
              & Yahoo!  & Istella-S & TianGong-ST \\ \midrule
    queries   & 28,719  & 32,968  & 3,449  \\
    documents & 700,153 & 3,406,167 & 333,813 \\
    features  & 700     & 220  & 33       \\
    relevance levels  & 5     & 5    & 5   \\
    avg. documents / query & 24.38 & 103.32 & 96.79 \\
    \bottomrule
    \end{tabular}
\end{table}

\label{sec:training_detail}

Similar to \cite{vardasbi2020inverse,ai2018unbiased}, we trained an MLP as our ranking model, where the implementation was the same as ULTRA framework \cite{ai2018ultra,ai2021unbiased}. To make fair comparisons, all the baselines and our model shared the same number of hidden units. The only difference in the ranking model between our model and baselines was the output dimensions. For our method, $\lambda$ was set to be $0.001$, the learning rate was tuned to be $0.05$, and the dimension was selected from $\{1,2,3,4,5\}$. We used two layers with sizes $\{256, 64\}$ and \textit{elu} activation as the implementation of the base model. We trained all these methods with a batch size of 256 and used AdaGrad to train the models. To ensure convergence, we train 15,000 epochs on the Yahoo! dataset, and 30,000 epochs on the Istella-S dataset. We used nDCG@$k$ $(k=1,3,5,10)$ as the performance metrics, and run each experiment for 8 times. We adopted the model with the best results based on nDCG@$10$ tested on the validation set and reported the average results testing on the test set.

\subsection{Discussion of relevance embedding}

\label{sec:discuss_relevance_emb}

From Figure \ref{fig:scatter_tiangong} and \ref{fig:scatter_trust}, we can see that the closer the true relevance is, the smaller the distance between them in the relevance embedding space. This is a good property since the distance relation can be kept approximately after linear transformation. That is to say, if we project them onto scalars, the close relevance embeddings should also have a close scalar distance, even if the base vector is inaccurate. It shows the robustness of the relevance embedding space.

Besides, one can observe that these relevance embeddings have a linear style. Why these points have this pattern is an open question. We tried to use PCA to find this direction and projected them onto it to see the performance. However, this method was not better than the method mentioned in this paper. This feature may be useful for finding better sorting methods in the future.

\end{document}